\documentclass[12pt]{iopart}
\usepackage{cite}
\usepackage{amssymb, amsfonts, bbm, amsthm, eqparbox}
\usepackage{array, float}
\usepackage{eufrak}
\usepackage{graphicx,epsfig}

\newcommand{\ket}[1]{\left\vert{#1}\right\rangle}

\newtheorem{theorem}{Theorem}[section]

\begin{document}

\title{Protecting Information Against Computational Errors and Quantum Erasures via Concatenation}

\author{Gilson O. dos Santos$^{1,3}$ and Francisco M. de Assis$^{2,3}$}

\address{$^{1}$ Federal Institute of Education, Science and Technology at Alagoas, Macei\'{o}, Brazil}
\address{$^{2}$ Department of Electrical Engineering, Federal University of Campina Grande, Brazil}
\address{$^{3}$Institute for Studies on Quantum Computation and Quantum Information (IQuanta).}
\ead{gilson.santos@ee.ufcg.edu.br, fmarcos@dee.ufcg.edu.br}
\begin{abstract}
In this work, we introduce a new concatenation scheme which aims at protecting information against the oc\-cur\-rence of both computational errors and quantum erasures. According to our scheme, the internal code must be a quantum loss-correcting code that does not perform measurements, while the external code must be a quantum error-correcting code. We illustrate the concatenation proposed with an example in which one qubit of information is  protected against the occurrence of two erasures and one computational error.\footnote{This work was presented at the IV Workshop-school on Quantum Computation and Information (WECIQ 2012), Fortaleza-CE, Brazil.}
\end{abstract}

\maketitle

\section{Introduction}
One of the most important obstacles in applications of quantum computing and quantum information processing is a phenomenon known as {\it decoherence} \cite{Shor1995,Bus2007}. The use of quantum systems efficiently in various applications of computing and information processing is conditioned to mitigate the effects of decoherence, which can be seen as a consequence of quantum entanglement between the system and the environment. One implication of the occurrence of decoherence is the loss of quantum information \cite{Bouw2000,NieChu2000}.

Effects caused by the decoherence in a quantum state can be characterized as being composed of two types of {\it changes}: ($i$) those that are consistent with the conditions established by Knill and Laflamme \cite{KnilLafl1997} who states that the changes are represented by Pauli matrices $\sigma_{X}, \sigma_{Y}$ and $\sigma_{Z}$ which operate in what is called {\it computational space} \cite{NieChu2000}, and due to that are known as {\it computational errors} \cite{Stace2009}; and ($ii$) those that lead the state encoded out of the computational space, causing loss of the qubit \cite{Lu2008}. One channel model describing this latter type of change is the {\it quantum erasure channel} (QEC), proposed by Grassl et al. \cite{Gras1997} who consi\-de\-red a situation in which the position of the erroneous (lost) qubits is known. 

Computational errors and quantum erasures are types of changes that can occur naturally due to the interaction between quantum systems and their environment. Given the impossibi\-li\-ty of ignoring the existence of such changes, there are reports in the li\-te\-ra\-tu\-re that highlight the practical importance of the development of codes that are capable of protecting against these two types of changes \cite{Duan2010,Dawson2006}.

Grassl {\it et al.} \cite{Gras1997} found that a quantum code can correct $t$ errors at unknown locations if the same code can correct $2t$ errors at known locations (a situation analogous to the QEC). However, although the location of an erroneous qubit is known, it is not always possible to guarantee its correction.  For example, when a qubit suffers a change that takes its state out of the computational space, causing the loss (erasure) of such qubit \cite{Lu2008}. The loss of qubits is frequent in practical implementations with multiple qubits, such as Josephson junctions \cite{Fazio1999}, neutral atoms in optical lattices \cite{Vala2005} and, most notoriously, in single photons that can be lost during processing or may be the loss attributed to the use of inefficient sources and/or inefficient detectors \cite{NLM2001, WaBa2007}. Therefore, it is necessary to develop schemes that address both the occurrence of computational errors and quantum erasures, which are quite recurrent changes in practical applications involving quantum information processing.


In general, it is possible to manipulate existing codes to build a new code suitable for a more general error model. One of the tricks we can use for this is the {\it concatenation} \cite{Forney1966}, which enables to creation of a new code using existing ones. Concatenation, in fact, is considered a basic method for constructing good error-correcting codes. Furthermore, many of binary codes considered to be asymptotically good are constructed by concatenation  \cite{Dumer1998, Li2009}.

The first known applications of concatenated codes in quantum error-correction were reported by \cite{KnillLaf1996, Gott1997}. In the quantum scenario, concatenated codes play a key role in fault-tolerant quantum computing  \cite{KnillLaf1996, KLZ1996, Zalka1996}, and in building good degenerate quantum error-correcting code \cite{Grassl2009, Fujita2006, Cafaro2011}. For example, Gottesman \cite{Gott1997} built a new code by concatenating the five qubit code with itself. The degenerate code of Shor \cite{Shor1995,Cafaro2011} can be constructed by concatenating the bit-flip code with the phase-flip code, both of three qubits.

Although concatenation has already been applied to va\-rious scenarios in quantum information processing, we found no references in the literature that  concatenate a {\it quantum error-correcting code} (QECC) with a {\it quantum loss-correcting code} (QLCC). 

The aim of this paper is to present a concatenation scheme able to protect the information against the occurrence of both computational errors and quantum erasures. In this scheme, a QECC is used as the external code, while the internal code should be a QLCC which does not perform measurement. If this construction is respected, the resulting concatenated code protects the information against both computational errors and quantum erasures. To illustrate, we present a complete example in which one qubit of information is protected against the occurrence of one computational error and two quantum erasures.

This paper is organized as follows. We stablish our main result in Sec. \ref{propconcat}, showing the concetanated code able to protect the information against both  computational errors and quantum erasures. In Sec. \ref{exconcatprop} we present  an example that illustrates the implementation of the proposed concatenation. Finally, in Sec. \ref{concconcat} we present our concluding remarks.

\section{Concatenation of quantum codes}\label{propconcat}

The {\it serial concatenation} (or plainly {\it concatenation}) intend to produce a new code from two existing code, called {\it internal }and {\it external} codes. The {\it external code} $[M,k]$ that encodes $k$ qubits into $M$ qubits (rate $R_{c}^{ext}=k/M$), denoted by $C^{ext}$, is given by $C^{ext}=(E^{ext},D^{ext})$; while the {\it internal code} $[N,M]$ that encodes $M$ qubits into $N$ qubits (rate $R_{c}^{int}=M/N$), denoted by $C^{int}$, is given by  $C^{int}=(E^{int}, D^{int})$ where $E$ and $D$ correspond to the encoding and decoding operations, res\-pectively \cite{Benedetto1996}. Rahn et al. \cite{Rahn2002} showed that the map for a concatenation of codes is given by the composition of the maps for the constituent codes. This description will be used in presenting the following concatenation scheme.

A state of $k$-logical qubits $\rho_{i}$ is first encoded using the external code $C^{ext}$, producing a state of $M$ qubits $E^{ext} [\rho_{i}]$. Each of these qubits is then encoded by the internal code, i.e., using the map $E^{int} \otimes E^{int} \otimes \ldots \otimes E^{int} = (E^{int})^{\otimes M}$ acting on $E^{ext} [\rho_{i}]$. These sequential encodings compose the map encoding of the concatenated code

\begin{equation}\label{efetivenc}
\overline{E} = (E^{int})^{\otimes M} \circ E^{ext}.
\end{equation}

The $M$-qubit sections in $E^{ext}[\rho_{i}]$ are called {\it blocks}, containing $N$ qubits each. After sending through the channel, a noisy process $\overline{\mathcal{N}}$ acts on the $MN$-qubits previously encoded.

A simple error-correction scheme coherently corrects each of the code blocks based on the internal code, and then corrects the entire register based on the external code. Taking this into account, the map of the decoding for the concatenated code is given by

\begin{equation}\label{efetivdec}
\overline{D} = D^{ext} \circ (D^{int})^{\otimes M} .
\end{equation}

We denote the code concatenated with this correction scheme by $Q_{c}$, where

\begin{equation}\label{scheconc}
Q_{c} = C^{ext} (C^{int}) = (\overline{E},\overline{D})
\end{equation}
Note that $Q_{c}$ is a code of $MN$-qubits.

If $Q_{c}$ must be a concatenation of a QECC with a QLCC, then we must take into account that the decoding map (\ref{efetivdec}) must be performed without collapsing the quantum state before the application of decoding $D^{ext}$. To do so, it is necessary that no measurement must occur during the decoding process ($(D^{int})^{\otimes M}$).



Considering that the occurrence of an erasure is, by definition, flagged and located in some way, so naturally we think in first correct this type of change. After that,  we check whether there have been changes that are not flagged. In our context, it means that we must first perform the decoding of the QLCC  and then the decoding of the QECC (see Fig. \ref{esquemaconcat}).

\begin{figure}[h]
\centering
\includegraphics[scale=0.45]{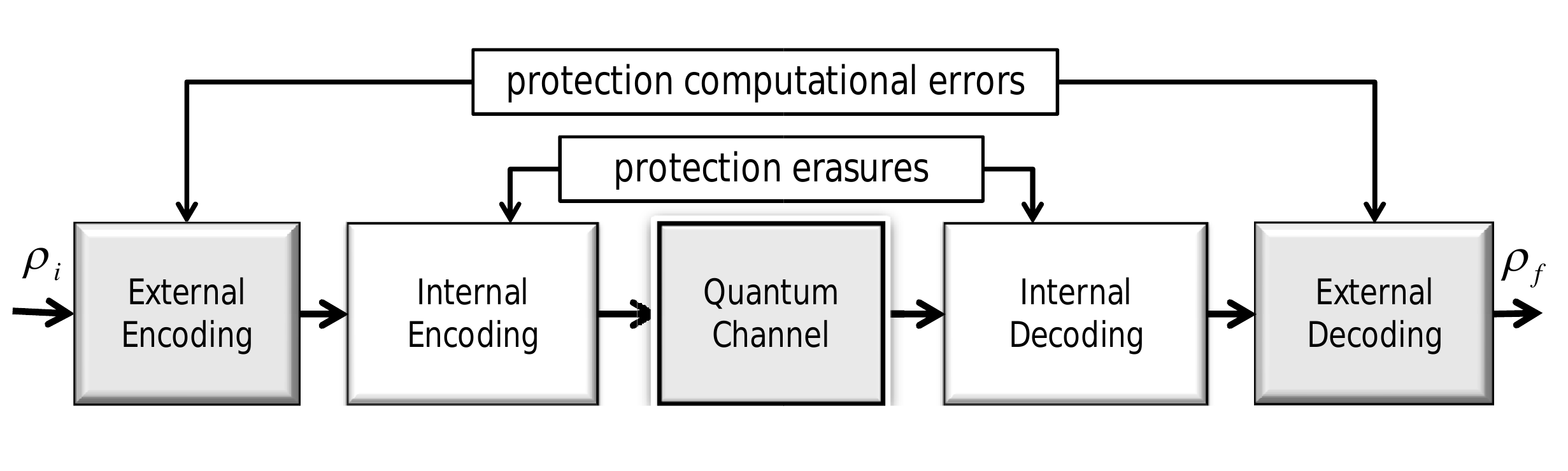}
\caption{\footnotesize Representation of the concatenation scheme for computational errors and erasures.}
\label{esquemaconcat}
\end{figure}


Taking into account the aforementioned considerations, the following theorem synthesizes the general idea of the proposed concatenation scheme.


\begin{theorem}\label{tconcapager}
Let $Q_{c}$ be a scheme as in (\ref{scheconc}) with $C^{int}$ as a QLCC (using unitary operations with auxiliary qubits and not performing measurements) capable to protect against up to $t$ erasures, and with $C^{ext}$ as a QECC capable to protect against up to $s$ computational errors. Then $Q_{c}$ is able to protect the information against the occurrence of up to $s$ computational errors and up to $t$ quantum erasures.
\end{theorem}

\begin{proof}
The demonstration will be held in a constructive manner. We assume that $Q_{c}$ as given in (\ref{scheconc}). Thus, the encoding process is as in (\ref{efetivenc}). Consider $C^{ext}$ as a QECC, able to protect the information against the occurrence of up to $s$ computational errors, and $C^{int}$ a QLCC (using unitary operations with auxiliary qubits and not performs measurement) which is capable of protecting against the occurrence information of up to $t$ quantum erasures. Assuming that you want to protect an arbitrary state $\rho_{i}$, then the result of the first step of encoding is given by $\rho ' = E^{ext}[\rho_{i}]$. As we assume that $C^{int}$ is a QLCC, then to complete the encoding process we apply $E^{int}$ in $\rho '$. This will result in

\begin{equation}\label{Qcenc}
\rho '' =E^{int}[\rho '].
\end{equation}

Having completed the $Q_{c}$ encoding, the state obtained in (\ref{Qcenc}) is sent through the quantum channel, being susceptible to the action of a noisy process. Admitting that occurred up to $t$ quantum erasures and up to $s$ computational errors, we have as a result the $\hat{\rho}$ state.

Proceeding with the $Q_{c}$ decoding, if we consider that the QLCC decoding is composed of a unitary operation with auxiliary qubits and not performing measurements, so the QLCC decoding is able to modify the received state such that its result is free of erasures, as established by \cite{Cai2004}. Thus, after the application of $D^{int}$ we obtain $\rho_{d}=D^{int}[\hat{\rho}]$, a state free of erasures which preserves the existing entanglement.

Now $\rho_{d}$ will be handled by $D^{ext}$, the decoding of the QECC. It follows that we must perform a detection procedure, making use of a measurement,  to verify if any computational error occurred and in which position \cite{NieChu2000,Beth1998}. This lets us know what should be the local operation to be performed to obtain the desired original state. Therefore, after applying $D^{ext}$, we have

\begin{equation}
\rho_{f}= D^{ext}[\rho_{d}].
\end{equation}

Thus, if $Q_{c}$ is a concatenated code in which $C^{int}$ is a QLCC (using auxiliary qbits and not performing measurements) and $C^{ext}$ is a QECC, then $Q_{c}$ is able to protect the information against the occurrence of both computational errors and quantum erasures. This concludes our proof.  
\end{proof}

The result of Theorem  \ref{tconcapager} shows that if the concatenation has $C^{int}$ as a code that protects against the occurrence of quantum erasures (which makes use of unitary operation with auxiliary qubits and does not perform measurements) and $C^{ext}$ as a code that protects against the occurrence of computational errors, so this concatenation scheme is able to perform the protection of information against these two types of changes.

One of the advantages of the proposed scheme is that in describing how to concatenate a QECC with a QLCC, we obtain the simultaneous protection of information against the occurrence of computational errors and quantum erasures.

It is important to emphasize that, although there are QECCs and QLCCs, the authors found no reports in the literature to describe a way to combine them. We should also point out that without the definition of the conditions that describe how to combine them, it would not be possible to develop a procedure that concatenates them.

In the next section there is an example that illustrates the implementation of the proposed concatenation.

\section{Example}\label{exconcatprop}

To illustrate how the information is protected by the concatenated code proposed in Theorem \ref{tconcapager}, we present an example in which $C^{ext}$ is the quantum $[[5,1,3]]$ code (obtained via a 3-regular graph), and $C^{int}$ is a QLCC (proposed by \cite{Santos2011}) that protects $5$-qubit of information against the occurrence of two quantum erasures.

Regarding the $[[5,1,3]]$ code, we have $n=5$, $k = 1$ and $d = 3$. According to~\cite{SchlWern2001}, the cardinalities of the sets $\mathcal{X}$ and $\mathcal{Y}$ are given by $|\mathcal{X}| = 1$ and $| \mathcal{Y} | = 5$. The graph representing this code is shown in Fig. \ref{figure:hexacodigo1}. 

\begin{figure}[h]
\centering
\includegraphics[scale=0.2]{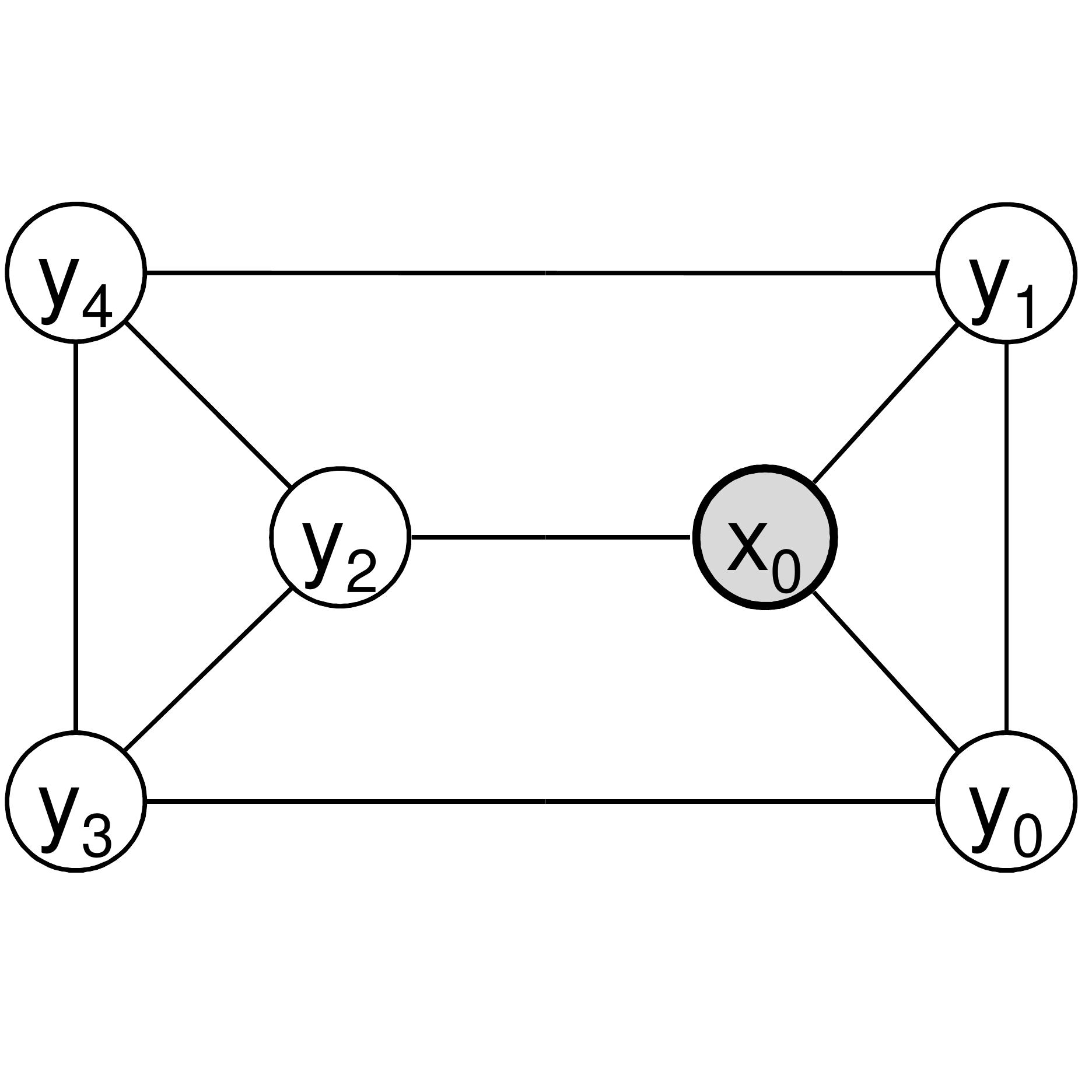}
\caption{\footnotesize $3$-regular graph for the $[[5, 1, 3]]$ code \cite{SchlWern2001}.}
\label{figure:hexacodigo1}
\end{figure}

The adjacency matrix corresponding to the graph in Fig. \ref{figure:hexacodigo1} is:

\begin{equation}\label{mtxgama}
\Gamma = 
\begin{array}{c}
x_{0} \\
y_{0} \\
y_{1} \\
y_{2} \\
y_{3} \\
y_{4} \\ 
\end{array}
\stackrel {
\begin{array}{l l l l l l}
 \ \ x_{0}   &   y_{0}   &   y_{1}   &   y_{2}   &   y_{3}   &   y_{4} \\ 
\end{array} }
{\left(
\begin{array}{l | r r r r r}
        0    & \  1  &  \ 1   &  \  1   & \ 0   & \  0   \\ \hline
       1    &  \ 0  &  \  1   &  \  0    & \ 1   & \  0  \\
       1    &  \ 1  &  \ 0   &  \ 0    & \ 0   &  \ 1  \\
       1    &  \ 0  & \  0   & \  0    & \  1  &  \ 1   \\
       0    &  \ 1  &  \ 0   &  \ 1    & \ 0   & \  1   \\
       0    &  \ 0  &  \  1  &  \  1   & \  1  & \  0  \\
\end{array}
\right). }
\end{equation}

We will consider here the field $\mathbb{F}_{2}=\{0,1\}$, that $x_{0}$ labels the input vertex, and that $y_{0}, y_{1}, y_{2}, y_{3}, y_{4}$ label the output vertices (see Fig. \ref{figure:hexacodigo1}). Hence, \mbox{$d^{\mathcal{X}}=(d^{x_{0}}) \in \mathbb{F}_{2}$} and $d^{\mathcal{Y}}= (d^{y_{0}}, d^{y_{1}}, d^{y_{2}}, d^{y_{3}}, d^{y_{4}}) \in \mathbb{F}_{2}^{5}$.

Following the proposition of Santos et al. \cite{Santos2013}, the encoding operator for this graph is given as follows (throughout this example we will omit the normalization factors):

\begin{eqnarray}\label{eqcod3r}
f(\vert v \rangle ) & = & \sum_{d^{x_{0}}=0}^{1} \left( \sum_{d^{y_{0}}=0}^{1} \sum_{d^{y_{1}}=0}^{1} \sum_{d^{y_{2}}=0}^{1} \sum_{d^{y_{3}}=0}^{1} \sum_{d^{y_{4}}=0}^{1}  e^{(\pi i) [ \gamma ] } \big\vert d^{y_{0}}d^{y_{1}}d^{y_{2}}d^{y_{3}}d^{y_{4}} \big\rangle \right) \nonumber \\
& & c (d^{x_{0}} ) , 
\end{eqnarray}\\
where $\vert v \rangle = \sum_{d^{x_{0}}=0}^{1} c (d^{x_{0}}) \vert d^{x_{0}} \rangle = c(0) \vert 0 \rangle + c(1) \vert 1 \rangle$ and 

\begin{equation}
\gamma =  \left\{ \frac{1}{2}\big[d^{x_{0}},d^{y_{0}},d^{y_{1}},d^{y_{2}},d^{y_{3}},d^{y_{4}} \big]^{T} \ \Gamma \ \left[\scriptsize \begin{array}{c} d^{x_{0}} \\ d^{y_{1}} \\ d^{y_{2}} \\ d^{y_{3}} \\ d^{y_{4}} \\  \end{array} \right] \right\} .
\end{equation}

Therefore, after the encoding, we obtain 

\begin{eqnarray}\label{estadografo}
 f(\vert v \rangle ) \rightarrow \ket{\psi} & = & \Big( \vert 00000  \rangle  + \vert 00001  \rangle  \ldots  - \vert 11110  \rangle \nonumber \\ 
& &  + \vert 11111  \rangle \Big) c(0) + \Big( \vert 00000  \rangle + \vert 00001  \rangle \nonumber \\
& &   + \vert 11110  \rangle + \ldots   - \vert 11111  \rangle \Big) c(1) .    
\end{eqnarray}

The above encoding is representing $E^{ext}$ of $Q_{c}$ and will protect one qubit of information against the occurrence of one computational error. The next step is the internal encoding $E^{int}$ of $Q_{c}$.

Since the resulting state of the encoding $E^{ext}$ is a state of five qubits, then the input state of $E^{int}$ has $k=5$ qubits. This way, $C^{int}$ will make use of $t = \lfloor 5/2 \rfloor = 2$ blocks of $5$ auxiliary qubits each, all initially in the state $\ket{0}$ \cite{Santos2011}. The resulting encoding operation for $E^{int}$ is given as follows:

\begin{equation}\label{enc5}
\ket{\psi}_{GHZ} = U_{enc}\big( \vert \psi \rangle_{(0)} \otimes \vert 00000 \rangle_{(1)}  \otimes \vert 00000 \rangle_{(2)}\big) ,
\end{equation}
where

\begin{eqnarray}\label{openc5}
U_{enc} & = & \prod_{d=0}^{2} \left( \prod_{i=1}^{4} C_{5(d),i(d)} \right)  \prod_{d=0}^{2} \Bigg( H_{5(d)} \Bigg) \nonumber \\
& & \prod_{d=1}^{2} \left( \prod_{i=1}^{5} C_{i(0),i(d)} \right) .
\end{eqnarray}\\
where $C_{x,y}$ representing a Controled-NOT (CNOT) operation and $H$ representing the Hadamard transform, respectively.

Rearranging the result of (\ref{enc5}), we obtain

\begin{eqnarray}\label{fiveancil}
\vert \psi \rangle_{GHZ} & = & \gamma_{0} \vert 0 \rangle_{L} + \gamma_{1} \vert 1 \rangle_{L} + \cdots  + \gamma_{30} \vert 30 \rangle_{L} + \gamma_{31}  \vert 31 \rangle_{L}, \nonumber \\
 & &
\end{eqnarray}\\
where $\gamma_{0} = c(0)  + c(1), \  \gamma_{1} = c(0)  + c(1),   \ldots , \ \gamma_{30} = - c(0)  + c(1), \  \gamma_{31} = c(0)  - c(1)$, and

\begin{eqnarray}\label{statefive}
\vert 0 \rangle_{L} & = & ( \vert 00000 \rangle + \vert 11111 \rangle )_{(0)} \otimes ( \vert 00000 \rangle + \vert 11111 \rangle )_{(1)} \nonumber \\
& & \otimes ( \vert 00000 \rangle + \vert 11111 \rangle )_{(2)} , \nonumber \\
\vert 1 \rangle_{L} & = & ( \vert 00000 \rangle - \vert 11111 \rangle )_{(0)}  \otimes ( \vert 00000 \rangle - \vert 11111 \rangle )_{(1)} \nonumber \\
& & \otimes ( \vert 00000 \rangle - \vert 11111 \rangle )_{(2)} , \nonumber \\
\vdots & \nonumber 
\end{eqnarray}
\begin{eqnarray}
\vert 30 \rangle_{L} & = & ( \vert 11110 \rangle + \vert 00001 \rangle )_{(0)} \otimes ( \vert 11110 \rangle + \vert 00001 \rangle )_{(1)} \nonumber \\
& & \otimes ( \vert 11110 \rangle + \vert 00001 \rangle )_{(2)} , \nonumber \\
\vert 31 \rangle_{L} & = & ( \vert 11110 \rangle - \vert 00001 \rangle )_{(0)} \otimes ( \vert 11110 \rangle - \vert 00001 \rangle )_{(1)} \nonumber \\
& & \otimes ( \vert 11110 \rangle - \vert 00001 \rangle )_{(2)}.
\end{eqnarray}

This completes the encoding $\overline{E}$ of $Q_{c}$, given in (\ref{efetivenc}).

We consider now the situation in which the encoded state $\ket{\psi}_{GHZ}$ undergoes the action of the environment that causes, for example, erasures in the qubit $1$ (phase-flip) of the index block ($0$) and in the qubit $5$ (bit-flip) of the index block $(1)$, as well as a computational error in the qubit $1$ (bit-flip) of the index block $(2)$. Thus, after these changes the resulting state will be $\vert e_{0} \rangle  \otimes \vert \psi \rangle_{GHZ} \rightarrow \vert \overline{\psi} \rangle_{GHZ}$ (where $\ket{e_{0}}$  is the initial state of environment) as follows:

\begin{eqnarray}\label{fivedesc2}
\vert \overline{\psi} \rangle_{GHZ} & = & \lambda_{0} \vert \overline{0} \rangle_{L} + \lambda_{1} \vert \overline{1} \rangle_{L} +  \cdots  + \lambda_{30} \vert \overline{30} \rangle_{L} + \lambda_{31} \vert \overline{31} \rangle_{L}, \nonumber \\
          & &
\end{eqnarray}
where

\begin{eqnarray}\label{fiveerror2}
\vert \overline{0} \rangle_{L} & = & ( \vert \overline{0}0000 \rangle - \vert \overline{1}1111 \rangle )_{(0)} \otimes ( \vert 0000\overline{1} \rangle + \vert 1111\overline{0} \rangle )_{(1)} \nonumber \\
& & \otimes ( \vert \b{1}0000 \rangle + \vert \b{0}1111 \rangle )_{(2)} , \nonumber \\
\vert \overline{1} \rangle_{L} & = & ( \vert \overline{0}0000 \rangle + \vert \overline{1}1111 \rangle )_{(0)} \otimes ( \vert 0000\overline{1} \rangle - \vert 1111\overline{0} \rangle )_{(1)} \nonumber \\
& & \otimes ( \vert \b{1}0000 \rangle - \vert \b{0}1111 \rangle )_{(2)} , \nonumber \\
\vdots & & \nonumber \\
\vert \overline{30} \rangle_{L} & = & ( - \vert \overline{1}1110 \rangle + \vert \overline{0}0001 \rangle )_{(0)} \otimes ( \vert 1111\overline{1} \rangle + \vert 0000\overline{0} \rangle )_{(1)} \nonumber \\
& & \otimes ( \vert \b{0}1110 \rangle + \vert \b{1}0001 \rangle )_{(2)} , \nonumber \\
\vert \overline{31} \rangle_{L} & = & ( - \vert \overline{1}1110 \rangle - \vert \overline{0}0001 \rangle )_{(0)} \otimes ( \vert 1111\overline{1} \rangle - \vert 0000\overline{0} \rangle )_{(1)} \nonumber \\
& & \otimes ( \vert \b{0}1110 \rangle - \vert \b{1}0001 \rangle )_{(2)} . 
\end{eqnarray}

Comparing the logical states of (\ref{fiveerror2}) with those of (\ref{statefive}), we can notice that for each qubit in which erasure occurred (with a bar at the top) of the logical states of (\ref{fiveerror2}), the index block $(2)$ was not affected by erasure. This block, however, has a different type of change (computational error), a bit-flip caused by the environment (indicated by a dash at the bottom).

According to the Theorem \ref{tconcapager}, we must start with the de\-co\-ding of the internal code $D^{int}$ then, only after that, we must perform the decoding of the external code $D^{ext}$.

The decoding $D^{int}$ of the QLCC is given via the restoring operation $\mathcal{R}$ \cite{Santos2011}, which is as follows:

\begin{eqnarray}\label{oprecov5}
\mathcal{R} & = & \bigg[ U_{rec}^{ 5, 1} \circ U_{dec} \Big( \vert \overline{\psi} \rangle_{GHZ} \otimes \ket{00000}_{(3)} \Big) \bigg] \nonumber \\
& & \bigg[ U_{rec}^{ 1, 0} \circ U_{dec} \Big( \vert \overline{\psi} \rangle_{GHZ} \otimes \ket{00000}_{(3)} \Big) \bigg]. 
\end{eqnarray} 

We first perform the $U_{dec}$ operator in  ($\vert \overline{\psi} \rangle_{GHZ} \otimes \ket{00000}_{(3)}$). Recall that this operator acts only in undamaged blocks. After that, we apply the recovery operators $U_{rec}^{a, b}$.

For this case, the $U_{dec}$ operator is given as follows:

\begin{eqnarray}\label{decod5}
U_{dec} & = & \prod_{d=0  (d \notin \{ 0, 1 \} )}^{2} \bigg( \prod_{i=1}^{5} C_{i(3),i(d)} \bigg) \nonumber \\
               &  & \prod_{d=0  (d \notin \{ 0, 1 \} )}^{2} \bigg( \prod_{i=1}^{5} C_{i(d),i(3)} H_{5(d)}  \prod_{i=1}^{4} C_{5(d),i(d)} \bigg) . \nonumber \\
               & &
\end{eqnarray}

Since the erasures occurred in the qubit of position $1$ of index block $(0)$ and in the qubit of position $5$ of the index block $(1)$, then the recovery operators for this situation are given as follows:

\begin{eqnarray}\label{recov510}
U_{rec}^{1, 0} & = &T_{1(3),5(3),4(0)} Z_{5(3), 4(0) } T_{1(3),5(3),4(0)} \nonumber \\
& & \prod_{i=1  (i \neq 1)}^{4} C_{i(3),i(0)} \prod_{i=1 (i \neq 1)}^{5} C_{1(3),i(0)} ,
\end{eqnarray}\\
where $T$ represents a Toffoli operation, $Z$ represents the $\sigma_{Z}$-Pauli controled operation, and in this case $\mathcal{W}_{(0)}= \{1, 2, 3, 4, 5 \} \setminus \{ 1 \} = \{ 2, 3, 4, 5 \}$ with $r = max_{r\neq 5}\{\mathcal{W}_{(0)} \}=4$; and also:

\begin{eqnarray}\label{recov551}
U_{rec}^{ 5, 1} & = & Z_{5(3), 4(1) }  \prod_{i=1}^{4} C_{i(3),i(1)},  
\end{eqnarray}\\
where $\mathcal{W}_{(1)}= \{1, 2, 3, 4, 5 \} \setminus \{ 5 \} = \{ 1, 2, 3, 4 \}$ and $r = max_{r\neq 5}\{\mathcal{W}_{(1)}\}= 4$. 

After applying the operator (\ref{decod5}) and the operators (\ref{recov510}) and (\ref{recov551}), respectively, we obtain

\begin{eqnarray}\label{recerror551}
\vert \overline{0} \rangle_{L} & = & ( \vert \overline{0}1111 \rangle - \vert \overline{1}0000 \rangle )_{(0)} \otimes ( \vert 1000\overline{0} \rangle + \vert 0111\overline{1} \rangle )_{(1)} \nonumber \\
& & \otimes  \vert 00000 \rangle_{(2)} \otimes  \vert \b{1}0000 \rangle_{(3)} , \nonumber \\
\vert \overline{1} \rangle_{L} & = & ( \vert \overline{0}1111 \rangle - \vert \overline{1}0000 \rangle )_{(0)} \otimes ( \vert 1000\overline{0} \rangle + \vert 0111\overline{1} \rangle )_{(1)} \nonumber \\
& & \otimes  \vert 00000 \rangle_{(2)} \otimes  \vert \b{1}0001 \rangle_{(3)} , \nonumber \\
\vdots & & \nonumber \\
\vert \overline{30} \rangle_{L} & = & ( \vert \overline{0}1111 \rangle - \vert \overline{1}0000 \rangle )_{(0)} \otimes ( \vert 1000\overline{0} \rangle + \vert 0111\overline{1} \rangle )_{(1)} \nonumber \\
& & \otimes  \vert 00000 \rangle_{(2)} \otimes  \vert \b{0}1110 \rangle_{(3)} , \nonumber \\
\vert \overline{31} \rangle_{L} & = & ( \vert \overline{0}1111 \rangle - \vert \overline{1}0000 \rangle )_{(0)} \otimes ( \vert 1000\overline{0} \rangle + \vert 0111\overline{1} \rangle )_{(1)} \nonumber \\
& & \otimes  \vert 00000 \rangle_{(2)} \otimes  \vert \b{0}1111 \rangle_{(3)} . 
\end{eqnarray}

Note that in (\ref{recerror551}), the block index $(0)$ as well as the index block $(1)$ now have the same form for all logical states. The system and the environment will be thus in the state:

\begin{eqnarray}\label{oprecov6}
& & \Big( \vert \overline{0}1111 \rangle - \vert \overline{1}0000 \rangle \Big)_{(0)} \otimes \Big( \vert 1000\overline{0} \rangle + \vert 0111\overline{1} \rangle \Big)_{(1)}  \nonumber \\
& & \otimes \Big( \ket{00000} \Big)_{(2)} \otimes \Big( \vert \overline{\psi} \rangle \Big)_{(3)}.
\end{eqnarray}

The received state is now free of quantum erasures. The next step is the external decoding $D^{ext}$ which will be used to verify and correct the occurrence of one computational error. 

The $\vert \overline{\psi} \rangle$ state in (\ref{oprecov6}) will be rewritten as follows, before the application of the procedures of the external decoding:
In order to facilitate understanding of the next step will change the way of representing the $\vert \overline{\psi} \rangle$ state. Thus, it is rewritten as follows

\begin{eqnarray}\label{estadografo4}
\vert \overline{\psi} \rangle & = & \Big( \vert \b{1}0000  \rangle  + \vert \b{1}0001  \rangle + \ldots - \vert \b{0}1110  \rangle \nonumber \\
& & + \vert \b{0}1111  \rangle \Big) c(0) + \Big( \vert \b{1}0000  \rangle  + \vert \b{1}0001  \rangle  + \nonumber \\
& &  \ldots + \vert \b{0}1110  \rangle - \vert \b{0}1111  \rangle \Big) c(1) .     
\end{eqnarray}

The first step in the decoding operation of the quantum graph codes $D^{ext}$ is the calculation of the error syndrome. To so, we apply the $\mathcal{T}$ decoding operator to the $\vert \overline{\psi} \rangle$ state \cite{Santos2013}. Since $\mathcal{T}$ is linear, it is applied to all basis states of $\vert \overline{\psi} \rangle$, resulting in:

\begin{eqnarray}\label{decgrafest5}
 \mathcal{T} (\vert \overline{\psi} \rangle ) & = & \Big[ \mathcal{T}(\vert \b{1}0000  \rangle )  + \mathcal{T}(\vert \b{1}0001  \rangle)  + \ldots - \mathcal{T}(\vert \b{0}1110  \rangle) \nonumber \\
 & & + \mathcal{T}(\vert \b{0}1111  \rangle) \Big] c(0) + \Big[ \mathcal{T}(\vert \b{1}0000  \rangle)    \nonumber \\
& & + \mathcal{T}(\vert \b{1}0001  \rangle)  + \ldots + \mathcal{T}(\vert \b{0}1110  \rangle) \nonumber \\
& &  - \mathcal{T}(\vert \b{0}1111  \rangle) \Big] c(1)  .    
\end{eqnarray}

\begin{figure}[htb]
               \centering
               \includegraphics[scale=0.2]{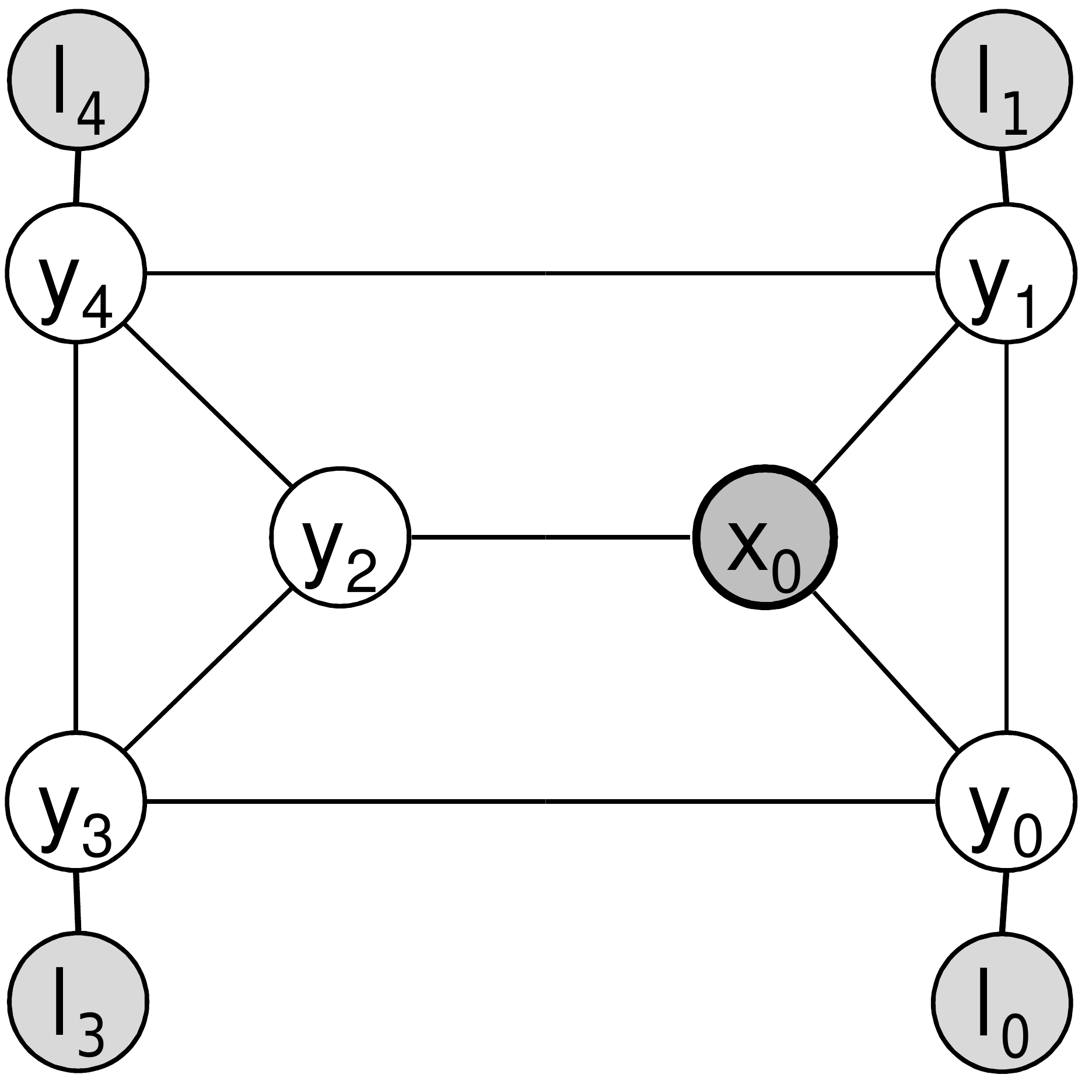}
               \caption{\footnotesize The 3-regular graph for the [[5,1,3]] code with syndrome vertices.}
               \label{figure:g3rsind}
\end{figure}

The representation of the graph decoding for $\vert \overline{\psi} \rangle$ is shown in Fig. \ref{figure:g3rsind}. According to this graph, the $\mathcal{T}$ decoding operator for each basis state is as follows:  

\begin{eqnarray}\label{invg3r}
\mathcal{T} \big( \big\vert d^{y_{0}}d^{y_{1}} d^{y_{2}}d^{y_{3}}d^{y_{4}} \big\rangle \big) & = & \sum_{d^{l_{0}}=0}^{1}\sum_{d^{l_{1}}=0}^{1}\sum_{d^{l_{3}}=0}^{1}\sum_{d^{l_{4}}=0}^{1} \sum_{d^{\widehat{x}_{0}}=0}^{1}  \nonumber \\
& & e^{-(\pi i) \left( \theta \right)} \big\vert d^{l_{0}}d^{l_{1}}d^{l_{3}}d^{l_{4}} d^{\widehat{x}_{0}} \big\rangle, \nonumber \\
& &
\end{eqnarray}\\
where $\theta = d^{\widehat{x}_{0}} d^{y_{0}} + d^{\widehat{x}_{0}} d^{y_{1}} + d^{\widehat{x}_{0}} d^{y_{2}} + d^{y_{0}} d^{y_{1}} + d^{y_{0}} d^{y_{3}} + d^{y_{1}} d^{y_{4}} + d^{y_{2}} d^{y_{3}} + d^{y_{2}} d^{y_{4}} + d^{y_{3}} d^{y_{4}} + d^{y_{0}} d^{l_{0}} + d^{y_{1}} d^{l_{1}} + d^{y_{3}} d^{l_{3}} + d^{y_{4}} d^{l_{4}}$.

Calculating  the expression (\ref{invg3r}) for each basis states of  in (\ref{estadografo4}), results in: 

\begin{eqnarray}\label{finaldecod}
\mathcal{T} (\vert \varphi \rangle ) & = &   \vert 01100  \rangle c(0) - \vert 01101 \rangle c(1) \nonumber \\
                                       & = & \vert 0110  \rangle \ket{0} c(0) - \vert 0110 \rangle \ket{1} c(1) \nonumber \\
                                      & = &   \vert 0 \rangle \vert 1 \rangle \vert 1 \rangle \vert 0 \rangle \Big( c(0) \vert 0 \rangle - c(1) \vert 1 \rangle \Big) . 
\end{eqnarray}

According to (\ref{invg3r}),  the syndromes qubits for this example correspond to the first four qubits in (\ref{finaldecod}). Measuring them in the computational basis, we obtain the {\it error syndrome} that, in this case, is equal to  $0110$. 

The next step is to check the table-lookup (Table \ref{tabsind3r}) to verify the type of computational error it corresponds and the action that must be performed in the fifth qubit. The error syndrome must be located in the Column 1. In this example, we can see that this syndrome indicates that a bit-flip occurred in the qubit $1$ (Column 2); that the general form of the state that must be recovered is $c(0)\ket{0} - c(1) \ket{1}$ (Column 3); and that the corresponding local correction operation is a phase-flip on the fifth qubit (Column 4).

In the final step, we must apply the local correction ope\-ra\-tion obtained from the table-lookup to restore the originally encoded state.

\begin{table}[htb] 
\caption{Table of syndromes: quantum code of 5 qubits via 3-regular graph}
\label{tabsind3r}
\begin{center}
\begin{tabular}{cccc} \hline \hline
Syndrome qubits           & Error (*) & State of      & Correction \\
$q_{1}q_{2}q_{3}q_{4}$ &            &  $q_{5}$ qubit  &  operation (*)  \\ \hline
0000      &     None   &   $c(0) \vert 0 \rangle + c(1) \vert 1 \rangle$         & None \\ 
0001      &        $S_{5}$    &   $c(0) \vert 0 \rangle + c(1) \vert 1 \rangle$      & None \\ 
0010      &        $S_{4}$    &   $c(0) \vert 0 \rangle + c(1) \vert 1 \rangle$      & None \\ 
0011      &        $B_{3}$    &   $c(0) \vert 0 \rangle - c(1) \vert 1 \rangle$       &  $S_{5}$  \\ 
0100      &        $S_{2}$    &   $c(0) \vert 0 \rangle + c(1) \vert 1 \rangle$      & None \\ 
0101      &        $B_{4}$    &   $c(0) \vert 1 \rangle + c(1) \vert 0 \rangle$      & $B_{5}$   \\ 
{\bf 0110} &      $\mathbf{B_{1}}$  & $\mathbf{c(0) \vert 0 \rangle - c(1) \vert 1 \rangle}$  & $\mathbf{S_{5}}$ \\ 
0111      &        $BS_{4}$    &   $- c(0) \vert 1 \rangle - c(1) \vert 0 \rangle$    & $SBS_{5}$ \\ 
1000      &        $S_{1}$    &   $c(0) \vert 0 \rangle + c(1) \vert 1 \rangle$      & None \\ 
1001      &        $B_{2}$    &   $c(0) \vert 0 \rangle - c(1) \vert 1 \rangle$        & $S_{5}$ \\ 
1010      &        $B_{5}$    &   $c(0) \vert 1 \rangle + c(1) \vert 0 \rangle$       & $B_{5}$ \\ 
1011      &        $BS_{5}$    &   $-c(0) \vert 1 \rangle - c(1) \vert 0 \rangle$    & $SBS_{5}$ \\ 
1100      &        $S_{3}$    &   $c(0) \vert 1 \rangle + c(1) \vert 0 \rangle$       & $B_{5}$ \\ 
1101      &        $BS_{2}$    &   $-c(0) \vert 0 \rangle + c(1) \vert 1 \rangle$    & $BSB_{5}$ \\ 
1110      &        $BS_{1}$    &   $-c(0) \vert 0 \rangle + c(1) \vert 1 \rangle$    &  $BSB_{5}$ \\ 
1111      &        $BS_{3}$    &   $-c(0) \vert 1 \rangle + c(1) \vert 0 \rangle$    & $BS_{5}$ \\ \hline
\end{tabular}
\end{center}
{\scriptsize 
(*) $B$ and $S$ denote the bit-flip and phase-flip operations, respectively; the subscript index  $n$ denotes the qubit to which the operations refers.
}
\end{table}

This example illustrates the encoding and decoding ope\-rations of the concatenated scheme proposed in the Theorem \ref{tconcapager}. It also shows explicitly the protection of the information against simultaneous multiple erasures and one computational error. 

It is important to emphasize, in a more general view, that the concatenated code in this example is able to protect against one computational error or two quantum erasures that occur in any of the three blocks, or against the simultaneous occurrence of two erasures and one computational error, since the computational error occur in the undamaged block (without erasure), as illustrated. Is should also be remarked that implementation of this code is feasible since it makes use of well-known quantum gates as CNOT, Hadamard, $\sigma_ {Z}$-Pauli, among others.

\section{Conclusion}\label{concconcat}
We introduced a new concatenation scheme for protecting information against the occurrence of both computational errors and quantum erasures. In the proposed scheme, the code must be a QLCC (which does not perform measurement), while the external code must be a QECC.

The approach presented in this paper concatenates two codes to provide protection against the errors for which they were originally designed. One advantage of this is that it benefits the QECC's and QLCC's that already exist, describing how to combine them.

The concatenation scheme proposed here can be especially interesting from the point of view of applications which consider a noise model to treat the occurrence of noise depolarization (computational error) and the occurrence of loss of photons (erasure). 

The proposed scheme is illustrated by an example in which, using a QECC obtained via a $3$-regular graph \cite{Santos2013} and a QLCC obtained as described in \cite{Santos2011}, we can protect one qubit information against the occurrence of two erasures and one computational error.

We show that one can concatenate a QECC with a QLCC to protect information against the occurrence of both computational errors and quantum erasures. However, it is important to emphasize that the general quantum decoding problem is NP-hard \cite{Hsieh2011}.

For future works, we suggest the concatenation of other QECC's and QLCC's existing in the literature aiming at finding a combination such that the concatenated code resulting would have a better rate than the combination shown here.



\section*{Acknowledgment}

The authors would like to thank A\'{e}rcio F. de Lima and Ello\'{a} B. Guedes for helpful discussions and comments.

\section*{References}

\addcontentsline{toc}{chapter}{Referências Bibliográficas}

\end{document}